\documentclass{article}

\PassOptionsToPackage{numbers, compress}{natbib}

\usepackage[main, final]{neurips_2026}

\usepackage[utf8]{inputenc} %
\usepackage[T1]{fontenc}    %
\usepackage[hidelinks]{hyperref}     %
\usepackage{url}            %
\usepackage{booktabs}       %
\usepackage{amsfonts}       %
\usepackage{nicefrac}       %
\usepackage{microtype}      %
\usepackage{xcolor}         %

\usepackage{color}
\usepackage{amsmath}
\usepackage{amssymb}
\usepackage{mathtools}
\usepackage{amsthm}
\usepackage{nicefrac}
\usepackage[linesnumbered,ruled,vlined,noend,algo2e]{algorithm2e}
\usepackage{enumitem}

\theoremstyle{plain}
\newtheorem{theorem}{Theorem}[section]
\newtheorem{proposition}[theorem]{Proposition}
\newtheorem{lemma}[theorem]{Lemma}

\newtheorem{definition}[theorem]{Definition}

\newtheorem{remark}[theorem]{Remark}\usepackage{color}
\usepackage{amsmath}
\usepackage{amssymb}
\usepackage{mathtools}
\usepackage{amsthm}
\usepackage{nicefrac}
\usepackage[linesnumbered,ruled,vlined,noend,algo2e]{algorithm2e}

\newcommand{\edit}[1]{#1}

\newcommand{\mindist}{d_{\rm min}}
\newcommand{\maxdist}{d_{\rm max}}
\newcommand{\dmin}{\mindist}
\newcommand{\dmax}{\maxdist}
\newcommand{\tarr}{t_{\rm arr}}
\newcommand{\tdel}{t_{\rm del}}
\renewcommand{\epsilon}{\varepsilon}
\newcommand{\eps}{\varepsilon}
\DeclareMathOperator*{\argmax}{argmax} %
\DeclareMathOperator*{\argmin}{argmin} %

\newcommand{\BO}[1]{O\!\left(#1\right)}

\newcommand{\delset}{\mathcal{D}}
\newcommand{\Ag}{\ensuremath{A^\gamma}}

\newcommand{\Rg}{\ensuremath{R^\gamma}}
\newcommand{\repg}{\ensuremath{rep^\gamma}}
\newcommand{\Let}[2]{#1 $\leftarrow$ #2}
\newcommand{\Agt}[1][]{\ensuremath{A^{\gamma, #1}}}
\newcommand{\Rgt}[1][]{\ensuremath{R^{\gamma, #1}}}
\newcommand{\Clu}{{D}}

\title{Dynamic k-center clustering with lifetimes}

\author{%
  Simone Moretti \thanks{equal contribution} \\
  KTH\\
  \texttt{\phantom{xxxxxxxxxx}simmor@kth.se\phantom{xxxxxxxxxx}} \\ %
    \And
    Paolo Pellizzoni \footnote[1]{} \\ 
    Harvard University \\
    \texttt{paolo\_pellizzoni@seas.harvard.edu} %
   \AND
  Andrea Pietracaprina \footnote[1]{} \\ 
  University of Padova \\
  \texttt{\phantom{xx}andrea.pietracaprina@unipd.it\phantom{xx}} %
  \And
  Geppino Pucci \footnote[1]{} \\ 
  University of Padova \\
  \texttt{\phantom{xxxxxx}geppino.pucci@unipd.it\phantom{xxxxx}} %
  \setcounter{footnote}{0}
}

\begin{document}

\maketitle

\begin{abstract}
  The $k$-center problem is a fundamental clustering variant with applications in learning systems and data summarization.
In several real-world scenarios, the dataset to be clustered is not static, but evolves over time, as new data points arrive and old ones become stale. To account for dynamicity, the $k$-center problem has been mainly studied under the sliding window setting, where only the $N$ most recent points are considered non-stale, or the fully dynamic setting, where arbitrary sequences of point arrivals and deletions without prior notice may occur. In this paper, we introduce the \textit{dynamic setting with lifetimes}, which bridges the two aforementioned classical settings by still allowing arbitrary arrivals and deletions, but making the deletion time of each point known upon its arrival. 
Under this new setting, we devise a deterministic $(2+\varepsilon)$-approximation algorithm with \smash{$\widetilde{O}(k/\varepsilon)$} amortized update time and memory usage linear in the number of currently active points. Moreover, we develop a deterministic $(6+\varepsilon)$-approximation algorithm that, under \smash{``tame''} update sequences, has \smash{$\widetilde{O}(k/\varepsilon)$} worst-case update time and heavily sublinear working memory. 
\end{abstract}

\section{Introduction}

Clustering, which is used to summarize and to group together data, is one of the most widely studied primitives in machine learning and data mining \citep{hennig2015handbook}. A common clustering task is $k$-center clustering, where, given a set of data points $X$ from a metric space $(\mathcal U, d)$, one needs to select $k$ centers from $X$ such that the maximum distance from any data point at hand to its closest center is minimized. This problem is well-known to be NP-hard \citep{gonzalez1985clustering}, so approximation algorithms are sought. 

Motivated by many real-world scenarios, recent years have witnessed a 
growing interest in devising dynamic clustering algorithms \citep{chan2018fully, bhattacharya2024fully, cohen2016diameter}, where the set of data points is not static, but rather changes dynamically, allowing new points to be inserted and others to be deleted. For example, the need for clustering dynamic data may arise for tracking stored products in a warehouse \citep{HichamAAO24, SrinivasanM99},
for analyzing the behaviors of active users in a web community \citep{MorrisonMHH21}, or studying evolving populations \citep{Gajer12}. %

Dynamic clustering focuses on the analysis of (unbounded) sequences of updates (i.e., insertions or deletions of points), focusing, at any time, on the instance defined by the \emph{active set} of points, namely, those that have been inserted and not yet deleted. 
The two main dynamic settings that allow both insertions and deletions are the \textit{sliding window} setting and the \textit{fully dynamic} setting.
In the sliding window setting, the active set consists of the $N$ most recent points, that is, whenever a new point is inserted, the oldest one is deleted  \citep{cohen2016diameter, borassi2020sliding}. The focus is on designing algorithms which use $o(N)$ memory, hence avoiding storing the entire active set. In the fully dynamic setting, both insertions and deletions are arbitrary: at any time, a new point may be inserted or an active point may be deleted  \citep{chan2018fully, bhattacharya2024fully}. This latter setting is more general and therefore more challenging than the former, but known approaches typically allow memory usage linear in the size of the active set, shifting the focus to approximation quality and fast update times.

In this paper,  we introduce a third setting, offering a middle-ground between the aforementioned two, which we refer to as the \textit{dynamic setting with lifetimes}. 
In this setting, when a point is inserted, an arbitrary deletion time is also specified, which defines the point's lifetime.  
This is more general than the sliding window setting, where the lifetimes are equal and fixed for all points, but it is less general than the fully dynamic one, where deletion times are not known in advance. 

The dynamic setting with lifetimes %
has several relevant real-world applications. For example, with reference to the aforementioned scenarios, items might arrive at a warehouse with a pre-assigned specific departure time, or subscriptions to a web community might have known expiration dates.
\edit{Another concrete example is provided by public TLS certificates recorded in Certificate Transparency logs. Indeed, when a certificate is recorded, it %
already specifies its expiration time, so the deletion time is known upon arrival~\cite{cooper2008internet, laurie2021certificate}.}

\paragraph{Our contribution}
In this paper, we tackle the $k$-center problem under the dynamic setting with lifetimes. In particular, we introduce:
\vspace{-0.5em}
\begin{itemize}
    \item a deterministic $(2+\eps)$-approximation algorithm with $\widetilde{O}(k/\eps)$ amortized update time\footnote{The notation $\widetilde{O}$ hides terms logarithmic in the relevant parameters.} and %
    memory usage linear in the number of currently active points;
    \item a deterministic $(6+\eps)$-approximation algorithm that under tame update sequences (see Definition~\ref{def:h-bounded}) has $\widetilde{O}(k/\eps)$ worst-case update time and $\widetilde{O}(k/\eps)$ working memory. 
\end{itemize}

A thorough comparison of the two algorithms with previous work is carried out in Subsections~\ref{sec:comparison-2} and \ref{sec:comparison-6}, respectively, showing that 
the knowledge of lifetimes yields improved performance guarantees over the state-of-the-art algorithms in the fully dynamic and sliding window settings.

\paragraph{Additional related work}

Several works have tackled k-center in the fully dynamic setting \citep{chan2024fully, bateni2023optimal}, 
focusing on variants with outliers \citep{chan2024fully, biabani2024improved}, 
on low-dimensional spaces \citep{goranci2021fully, pellizzoni2023fully, gan2024fully}, 
on minimizing recourse \citep{lkacki2024fully, bhattacharyaalmost, forster2025dynamic}, 
and on graphs \citep{cruciani2024dynamic}.
In the sliding window model, \citet{cohen2016diameter} tackled k-center while subsequent works addressed variants \citep{pellizzoni2022k, wang2023fair, ceccarello2025fair}, focused on low-dimensional spaces \citep{pellizzoni2022adaptive} or related clustering problems \citep{borassi2020sliding}.
Dynamic settings akin to the one \textit{with lifetimes} have been addressed recently, including the \textit{deletions-look-ahead} model \citep{PengR22}, the \textit{general expiration streaming} model \citep{blank2026expiration}, and models with predicted deletion times \citep{LiuV24}. More details are in Appendix~\ref{app:akin}.

\paragraph{Organization of the paper} 
Section~\ref{sec:prelim} defines preliminary technical notions and formally specifies the $k$-center problem under the dynamic setting with lifetimes. Sections~\ref{sec:2approx} and \ref{sec:6approx}  present our algorithms and their analyses, providing, for each algorithm, a detailed comparison with previously known ones. Section~\ref{sec:concl} concludes with some final remarks. Due to space constraints, 
additional considerations on related and future work, some pseudocodes, and some proofs are deferred to the appendix. 

\section{Preliminaries}
\label{sec:prelim}

\paragraph{$k$-center}
Consider a metric space $(\mathcal U,d)$ and a finite set $X \subseteq \mathcal U$. 
For any point $p \in X$ and any subset  $C \subseteq X$, 
we define $d(p,C) = \min_{q \in C} d(p,q)$,
and for any $S\subseteq X$ we denote the \emph{radius of $C$ with respect to $S$} as
$
r_C(S) = \max_{p \in S} d(p,C).
$ 

The \emph{$k$-center} problem, for a given positive integer $k$,  requires
to find a subset $C \subseteq X$ such that $|C| \leq k$ which minimizes
$r_C(X)$, i.e., $C^* = \arg \min_{C \subseteq X, |C| \leq k} r_C(X)$. 
We denote the radius of the optimal solution by $r_{k}^*(X) = r_{C^*}(X)$.
The points of the set $C$ are called \emph{centers}.  
Indeed, $C$ induces a partition of $X$ into $|C|$ clusters by assigning each point to its closest center (with ties
broken arbitrarily).  
In general metric spaces, the problem is NP-hard, and it is impossible to achieve an approximation factor $2-\epsilon$, for any fixed $\epsilon >0$, unless P = NP \citep{hsu1979easy, gonzalez1985clustering}.

\paragraph{The dynamic setting with lifetimes}

Let $X$ be a dataset whose subsets represent instances of a given
optimization problem\footnote{Although we define the dynamic setting with lifetimes with respect to optimization problems, several computational problems can be cast in this setting.}. In the \emph{dynamic with lifetimes setting}, each element $p \in X$ is endowed with an \emph{arrival time} $\tarr(p) \in \mathbb{N}$ and a \emph{deletion time} $\tdel(p) \in \mathbb{N}$, with $\tarr(p) < \tdel(p)$, and we say that $p$ \emph{expires} (i.e., its lifetime ends) at time $\tdel(p)$. 
We assume 
that for each time $t \geq 0$, there exists at most one element $p \in X$ with $\tarr(p) = t$. 
At any time $t \geq 0$, we say that an element $p$ is \emph{active} if
$\tarr(p) \leq t < \tdel(p)$, and define, accordingly,
the \emph{active set} $X_t$ as 
\[
X_t = \{p \in X \; : \;  \tarr(p) \leq t < \tdel(p)\}.
\]
Also, we assume that an oracle communicates the arrivals of elements from $X$, ordered by their arrival time.
\edit{The oracle reveals a point $p$ and its deletion time $t_{\mathrm{del}\mkern-1mu}(p)$ only at time $t_{\mathrm{arr}\mkern-1mu}(p)$ and provides no information about future arrivals.}
An algorithm in this dynamic setting maintains a suitable representation of $X_t$, by running a primitive $\textsc{Update}(p)$ upon the arrival of a new element $p$, and, at any time $t$, is able to compute a solution 
with respect to instance $X_t$ by running
a primitive $\textsc{Query}(t)$.

The goals that the algorithm should pursue are:
\begin{enumerate}
\item 
Good approximation of the optimal solution for $X_t$;
\item 
Query time sublinear in $|X_t|$;
\item 
Update time sublinear in $|X_t|$ (possibly amortized);
\item Memory requirements at most linear in $|X_t|$ (\textit{Optionally}: memory sublinear in $|X_t|$).
\end{enumerate}

Note that the main requirements are shared with both the fully dynamic and sliding window settings. The optional one, typical of the sliding window but not of the fully dynamic setting, makes the algorithm design more challenging. 
Under this requirement, our setting is analogous to the \textit{general expiration streaming model}, which was recently introduced in \citet{blank2026expiration} as an extension of the streaming model, thus focusing mainly on memory requirements.

\paragraph{Dynamic $k$-center with lifetimes}
Let $X$ be a set of points with arrival and deletion times, as described above. 
The dynamic $k$-center with lifetimes problem on $X$ requires that, at any given time $t$,
a set of $k$ centers $C_t \subseteq X_t$ be determined, whose radius $r_{C_t}(X_t)$ is close to the optimal one $r_k^*(X_t)$ for $X_t$. An algorithm for this problem consists of 
an $\textsc{Update}$ procedure, which  deals with insertions and deletions, so to maintain a suitable representation of the active set $X_t$, and a  
$\textsc{Query}$ procedure, which, at any given time $t$,
returns a solution $C_t$ for $X_t$. The approximation ratio of the algorithm is the maximum of $r_{C_t}(X_t) / r_k^*(X_t)$ over any set $X_t$. 

Let $\mindist$ and $\maxdist$ be the
minimum and maximum distances between any two distinct points of $X$, and define the \emph{aspect ratio} of $X$ as $\Delta = \maxdist/\mindist$. As standard in most previous works, our algorithms will assume that  $\mindist$ and $\maxdist$ %
are known. 

\section{A $(2 +\epsilon)$-approximation algorithm} \label{sec:2approx}

This section presents a deterministic algorithm for dynamic $k$-center with lifetimes, which attains an almost-tight $(2+\eps)$ approximation ratio, requiring memory linear in $|X_t|$ and $\BO{(\log{\Delta}/\eps) \cdot k + \max_t\log|X_t|}$ amortized update time.
\edit{We assume that $|X_t|> k$ whenever a query is made. Queries when $|X_t| \leq k$ can be handled separately by returning the entire active set.}

The algorithm's $\textsc{Update}$ and $\textsc{Query}$ procedures, described later in this section, crucially exploit a data structure based on several guesses for the optimal radius. Specifically, since the optimal radius clearly falls in $[\dmin,\dmax]$, we define the following set of guesses:
\begin{equation}\label{eq:gamma}
\Gamma = \left\{
(1+\beta)^i : 
\lfloor \log_{1+\beta} \dmin \rfloor
\leq i \leq
\lceil \log_{1+\beta} \dmax \rceil
\right\},
\end{equation}
where $\beta=\eps/2$, and $\eps>0$ is the user-defined approximation tolerance.
At a high level, for each $\gamma \in \Gamma$ we dynamically maintain a data structure, akin to the one used in \cite{chan2018fully}, consisting of the sets:
\begin{itemize}
	\item $C^\gamma$ = $(c^\gamma_1, \dots, c_{\ell^\gamma}^\gamma)$, a list of $\ell^\gamma \leq k$ centers, such that for each $1 \leq i \neq j \leq \ell^\gamma$, we have $d(c^\gamma_i, c^\gamma_j) > 2\gamma$;
	\item $\Clu^\gamma$ = $(\Clu_1^\gamma, \dots, \Clu_{\ell^\gamma}^\gamma)$, a list of $\ell^\gamma \leq k$  disjoint sets of points, such that for each
    $1 \leq i \leq \ell^\gamma$ and $x \in \Clu_i^\gamma$, we have $d(x, c_i^\gamma) \leq 2 \gamma$;
	\item $U^\gamma$, a set of unclustered points such that, for each $x \in U^\gamma$ and $1 \leq i \leq \ell^\gamma$, we have $d(x, c_i^\gamma) > 2 \gamma$. Moreover $U^\gamma \neq \emptyset$ implies $\ell^\gamma = k$.
\end{itemize}

\edit{Intuitively, for a given guess $\gamma$, the sets $D^\gamma_i$'s are clusters of radius $2\gamma$ centered at the points $c^\gamma_i$'s. If $\gamma$ is too small, we will not be able to cover the entire active set with at most $k$ clusters, and the remaining uncovered points are thus stored in the set $U^\gamma$.}

To keep track of which points need to be deleted from the active set, we maintain a copy of all points in a priority queue $Q$ indexed on their deletion time, with ties broken by the arrival time to guarantee a unique ordering. We implicitly assume the same tie-breaking in proofs,
\edit{ so every comparison involving deletion times
is understood with respect to the resulting total order. Given this total order, with a slight abuse of notation, in this section we implicitly reindex all
operations by consecutive integers in execution order, so that
$t_{\mathrm{del}}(p)-1$ denotes the step immediately before the deletion
of $p$.}

The $\Clu_i^\gamma$'s and $U^\gamma$'s are implemented as doubly-linked lists, and each entry in $Q$ has a pointer to the entries of the $\Clu_i^\gamma$'s and $U^\gamma$'s corresponding to the same point, so to allow these latter entries to be identified in $\BO{1}$ time. By construction, if $U^\gamma = \emptyset$, then $C^\gamma$ provides a solution for $X_t$ of radius $\leq 2 \gamma$, while, if $U^\gamma \neq \emptyset$, then $C^\gamma$
plus any point of $U^\gamma$ yields a set of $k+1$ points of $X_t$ at pairwise distance $> 2 \gamma$, which is a witness that no solution of radius $\gamma$ for $X_t$ exists.

\paragraph{Update procedure}
Procedure $\textsc{Update}_{2+\eps}(p,t)$ (see Algorithm~\ref{alg:update-2eps} for the pseudocode) is called when a new point $p$ arrives (in this case $t=\tarr(p)$), or at any time $t$ when expired points, if any, need to be deleted from the data structure (in this case, $p = \textsc{null}$).
First, for each $q \in Q$ with $\tdel(q) \leq t$, $q$ is
removed from $Q$ and procedure $\textsc{Delete}(q, \gamma)$ is invoked
for each $\gamma \in \Gamma$. Then, if $p \neq \textsc{null}$, procedure $\textsc{Insert}(p, \gamma)$ is invoked for each $\gamma \in \Gamma$, and $p$ is added to the priority queue $Q$.

$\textsc{Insert}(p,\gamma)$ (see Algorithm~\ref{alg:insert} for the pseudocode) inserts $p$ in the cluster of $\Clu^\gamma$ of lowest index whose center is at distance at most $2\gamma$ from $p$. If no such cluster exists, then if $\Clu^\gamma$ comprises $\ell^\gamma < k$ clusters, then the procedure creates a new cluster by making $p$ its center, otherwise, it inserts $p$ in $U^\gamma$.
At the end of $\textsc{Insert}(p,\gamma)$ a reclustering procedure, described later, is called. %

\begin{figure}[t]
\noindent
\begin{minipage}[t]{0.48\textwidth}
\begin{algorithm2e}[H]
\vspace{0.03cm}
\SetAlgoLined
    \While{\edit{$Q \neq \emptyset$} and $\min_{q\in Q} \tdel(q) \leq t$}{
        Remove \edit{$q = {\rm argmin}_{q\in Q} \tdel(q)$} from $Q$ \\
        \For{$\gamma \in \Gamma$}{$\textsc{Delete}(q, \gamma)$}
    }
    \If{$p \neq \textsc{null}$}{
    \For{$\gamma \in \Gamma$}{
        $\textsc{Insert}(p, \gamma)$
    }
	Insert $p$ in $Q$
    }
\vspace{0.07cm}
\caption{$\textsc{Update}_{2+\eps}(p,t)$} \label{alg:update-2eps}
\end{algorithm2e}
\end{minipage}%
\hfill
\begin{minipage}[t]{0.48\textwidth}
\begin{algorithm2e}[H]
\SetAlgoLined
    	\For{$i = 1, \dots \ell^\gamma$}{
		\If{$d(p, c_i^\gamma) \leq 2\gamma$}{
			Insert $p$ in $\Clu_i^\gamma$, and \textbf{goto} line~\ref{line:reclustering-insert} \\
        }
	}
    \If{$\ell^\gamma < k$}{
		Insert $p$ in $C^\gamma$ and in $\Clu^\gamma_{\ell^\gamma+1}$ \label{line:insert-center-addition}\\
            $\ell^\gamma \leftarrow \ell^\gamma+1$\\
	}
	\lElse{Insert $p$ in $U^\gamma$}
    $\textsc{Reclustering}(\gamma)$ \label{line:reclustering-insert}
\caption{$\textsc{Insert}(p, \gamma)$} \label{alg:insert}
\end{algorithm2e}
\end{minipage}
\vspace{-1em}
\end{figure}

$\textsc{Delete}(p,\gamma)$, where $p$ is an expired point, works as follows (see Algorithm~\ref{alg:delete}). If $p$ is not a center, it is simply removed from its cluster in $\Clu^\gamma$. If $p$ is a center $c_i^\gamma$, the procedure deletes $c_i^\gamma$ and tries to recluster all points of $\Clu_i^\gamma$ in clusters with index $>i$. Intuitively, each of these points is inserted in the cluster where it would have been inserted if $c_i^\gamma$ had not existed. Then, all cluster indexes greater than $i$ are decremented by 1, making them contiguous again, and, if $U^\gamma \neq \emptyset$, the point of $U^\gamma$ with the latest deletion time becomes the center of a new cluster. Finally, the reclustering procedure is called.

Note that the expensive part of the deletion procedure, namely, re-assigning points and creating a new cluster, only happens if a cluster center is deleted. In order to limit these costly operations, we aim at selecting centers with largest deletion times. In particular, this is enforced when a new center from $U^\gamma$ must be selected (Lines~\ref{line:del_create_cluster}-\ref{line:del_insert_in_C_l} of $\textsc{Delete}(p, \gamma)$), and when the reclustering is performed at the end of 
$\textsc{Insert}$ and $\textsc{Delete}$, which is described next. 

For a given $\gamma \in \Gamma$, we define as \textit{persistent} those clustered points which expire after their respective
clusters' centers, and \textit{vanishing} those which either expire earlier than their clusters' centers  or are unclustered. Note that only persistent points can change cluster because of the deletion of their center, as vanishing points expire before the deletion of their center. 
We partition each cluster $\Clu^\gamma_i$ into persistent and vanishing points as $P_i^\gamma := \{x\in \Clu^\gamma_i \colon \ \tdel(c^\gamma_i) < \tdel(x)\}$ and $V_i^\gamma := \{x\in \Clu^\gamma_i \colon \ \tdel(c^\gamma_i) \geq \tdel(x)\}$.
In practice, we maintain two counters per cluster for $|P_i^\gamma|$ and $|V_i^\gamma|$, which can be updated in $\BO{1}$ time for each insertion and deletion.

Procedure $\textsc{Reclustering}(\gamma)$ (see Algorithm~\ref{alg:reclustering}), first searches for the smallest index $i$ such that 
\vspace{-0.15cm}
\begin{equation} \label{eq:size_balance}
    \sum_{j=i}^{\ell^\gamma} |P_j^\gamma| > |U^\gamma| + \sum_{j=i}^{\ell^\gamma}|V_j^\gamma|.
\end{equation}
\vspace{-0.4cm}

If such $i$ exists, then it reclusters the points of 
$U^\gamma \cup (\cup_{j = i}^{\ell^\gamma} \Clu_j^\gamma)$ into 
$\ell^\gamma-i+1$ clusters,
selecting, for each such new cluster, a center with latest deletion time. 
Note that, after reclustering, some points might remain unclustered, but 
the newly created clusters comprise only vanishing points.

\begin{figure}[t]
\noindent
\begin{minipage}[t]{0.46\textwidth}
\begin{algorithm2e}[H]
\small
\SetAlgoLined
    \If{$p \not\in C^\gamma$}{
		Remove $p$ from $U^\gamma$ or $\Clu^\gamma$ \\
        \textbf{goto} line~\ref{line:reclustering-delete} \\
    }
	Let $i$ be such that $p = c_i^\gamma$ \\
    \For{$x\in \Clu_i^\gamma$ {\rm with} $x \neq p$}{ \label{line:deletion-reclustering-for}
        \For{$j = i+1,\dots, \ell^\gamma$}{
            \If{$d(x,c_j^\gamma)\leq2\gamma$}{
                Insert $x$ in $\Clu_j^\gamma$, and \textbf{break} \\ \label{line:del_insert_in_C_j}}
        }
        \If{$x$ was not inserted}{
            \If{$\ell^\gamma < k$}{
    		  Insert $x$ in $C^\gamma$ and in $\Clu_{\ell^\gamma+1}^\gamma$\;
              $\ell^\gamma\leftarrow \ell^\gamma+1$ \label{line:delete-center-addition}}
        	\lElse{Insert $x$ in $U^\gamma$\label{line:delete-insU}}
        }
    }
    Remove $c_i^\gamma$ and $\Clu_i^\gamma$\;
    $\ell^\gamma \leftarrow \ell^\gamma-1$\;
    Rename indexes of $C^\gamma, \Clu^\gamma$ to be contiguous\label{line:del_relabeling}\\
    \If{$U^\gamma \neq \emptyset$}{ \label{line:del_create_cluster}
        $\ell^\gamma \leftarrow \ell^\gamma+1$\;
        Pick $u\in U^\gamma$ s.t. $\tdel(u)$ is maximal\;
        Remove $u$ from $U^\gamma$, insert in $C^\gamma\!, \Clu^\gamma$ as $c_{\ell^\gamma}^\gamma$\!\!\!\\
        \For{$x\in U^\gamma$}{
            \If{$d(x,u)\leq 2\gamma$}{
                Remove $x$ from $U^\gamma$, insert in $\Clu_{\ell^\gamma}^\gamma$\!\!\! \label{line:del_insert_in_C_l}
            }
        }
    }
	  $\textsc{Reclustering}(\gamma)$ \label{line:reclustering-delete}
\caption{$\textsc{Delete}(p, \gamma)$} \label{alg:delete}
\end{algorithm2e}
\end{minipage}%
\hfill
\begin{minipage}[t]{0.512\textwidth}
\begin{algorithm2e}[H]
\small
\SetAlgoLined
$i \leftarrow $ min index $h$ s.t.
$\sum_{j=h}^{\ell^\gamma}\! |P_j^\gamma| > |U^\gamma| + \sum_{j=h}^{\ell^\gamma}\!|V_j^\gamma|\!\!\!\!\!$\\
\If{no such $i$ exists}{\Return}
$U' \leftarrow U^\gamma \cup (\cup_{j = i}^{\ell^\gamma} \Clu_j^\gamma)$ \label{line:reclustering_cluster_joining}\\
    Remove centers $c_j^\gamma$ and clusters $\Clu_j^\gamma$ $\forall j \geq i$\;
    $\ell^\gamma \leftarrow i-1$\;
    \For{$j = i, \dots, k$}{
        \If{$U' = \emptyset$}{
        \textbf{break}
        }
        $\ell^\gamma \leftarrow j$\;
        Pick $u\in U'$ s.t. $\tdel(u)$ is maximal\;
        Remove $u$ from $U'$, insert in $C^\gamma, \Clu^\gamma$ as $c_j^\gamma$\;
        \For{$x\in U'$}{
            \If{$d(x,u)\leq 2\gamma$}{
                Remove $x$ from $U'$, insert in $\Clu_j^\gamma$
            }
        }
    }
    $U^\gamma\leftarrow U'$
\caption{$\textsc{Reclustering}(\gamma)$} \label{alg:reclustering}
\end{algorithm2e}
\end{minipage}
\vspace{-1.0em}
\end{figure}

The following remark 
shows that the reclustering performed after each invocation of $\textsc{Insert}$ or $\textsc{Delete}$ 
is crucial to obtain sublinear amortized update time, 
\edit{as without it a linear number of updates may require a quadratic number of operations.}

\begin{remark}
\label{remark:sequence}
    There exists a sequence of $\BO{n}$ calls to $\textsc{Update}_{2+\eps}$ which, without the use of $\textsc{Reclustering}$, would perform a total of $\Theta(n^2)$ operations.
\end{remark}

\textbf{Query procedure} $\ $
At any time $t$, an approximate solution to $k$-center for the active set $X_t$ can be computed calling $\textsc{Query}_{2+\eps}(t)$ (see Algorithm~\ref{alg:query-2} for the pseudocode). 
We assume that when $\textsc{Query}_{2+\eps}(t)$ is invoked, the data structure contains no expired points. 
This can be ensured by checking that the minimum deletion time of a point in $Q$ is $>t$, and, if this is not the case,  
invoking $\textsc{Update}_{2+\eps}(\textsc{null},t)$ before $\textsc{Query}_{2+\eps}(t)$.
The procedure simply
returns the set $C^{\gamma^*}$, with 
$\gamma^* := \min_{\gamma \in \Gamma} \{ \gamma \colon \ U^{\gamma} = \emptyset \}$, as a set of centers.

The following subsections analyze the approximation ratio, and
the time and memory requirements of the above algorithm. For convenience, we define a \textit{time step} as one call of \textsc{Insert} or \textsc{Delete}. Since 
 arrival times are distinct and $Q$ breaks ties for deletion times based on arrival times, we can 
assume, without loss of generality, that each call corresponds to a unique time step. For any time step $t \geq 0$, we use $X_t$ to denote the active set, and $C^{\gamma, t}, \Clu^{\gamma, t}, U^{\gamma, t}, Q^t$ to denote the various components of the data structure,
\textit{after} the $t_{\text{th}}$ call of \textsc{Insert} or \textsc{Delete}. Time step
$t=0$ refers to the initial phase of the algorithm, when nothing is yet inserted. %
We may drop $t$ %
when clear from the context.

\subsection{Approximation}

First, we prove that the algorithm maintains the following invariant. This can be then used to prove the approximation guarantee for $\textsc{Query}_{2+\eps}$.

\begin{lemma}
\label{lem:invariant}
For any time step $t$, the following properties hold for each $\gamma \in \Gamma$, and $x \in X_t$:
    \begin{enumerate}
        \item 
        \label{invariant-c1}
        if $d(x, C^{\gamma, t})> 2\gamma$, then $x \in U^{\gamma, t}$; 
        \item if $d(x, C^{\gamma, t}) \leq 2\gamma$, then $x \in \Clu^{\gamma, t}_i$ with $i = \argmin_{j=1, \dots \ell^{\gamma, t} }\{ c_j \in C^{\gamma, t} \colon \ d(x, c_j) \leq 2\gamma \}$.
    \end{enumerate}
    Moreover, if $U^{\gamma, t} \neq \emptyset$ then $|C^{\gamma, t}| = k$.
\end{lemma}

\edit{At a high level, the invariant states that all points that are not within $2\gamma$ of a center in $C^\gamma$ are maintained in $U^\gamma$ (which can happen only if there are already $k$ centers), and that if a point is within $2\gamma$ of a center in $C^\gamma$, then it is assigned to the cluster with the lowest index. }

\begin{theorem}
    For any time step $t$, the solution returned by $\textsc{Query}_{2+\eps}(t)$ is a $(2+\epsilon)$-approximation for the $k$-center problem on the active set $X_t$.
\end{theorem}
\begin{proof}
Recall that the algorithm returns $C^{\gamma^*, t}$, with
$\gamma^* = \min_{\gamma \in \Gamma} \{ \gamma \colon \ U^{\gamma, t} = \emptyset \}$. If $\gamma^*  \leq  \dmin$, the theorem trivially holds. Suppose $\gamma^* > \dmin$,
and let $\gamma' = \gamma^*/(1+\beta) \in \Gamma$. 
Clearly, $U^{\gamma', t} \neq \emptyset$, hence $|C^{\gamma', t}| = k$. Let $W = C^{\gamma', t} \cup \{ x \}$,
for an arbitrary $x \in U^{\gamma', t}$. 
By construction and by Lemma~\ref{lem:invariant} we have
that
$\forall p \neq q \in W$, $d(p, q) > 2\gamma'$. Moreover, since $|W|=k+1$, we have that $r_k^*(X_t) \geq \min_{p \neq q \in W} d(p, q)/2 > \gamma'= \gamma^*/(1+\beta)$, as at least two distinct points of $W$ are in the same optimal cluster. By
Lemma~\ref{lem:invariant}, 
$\max_{x \in X_t} d(x, C^{\gamma^*, t}) \leq 2 \gamma^* < 2 (1+\beta) r_k^*(X_t) \leq (2+\eps) r_k^*(X_t)$.
\end{proof}

\subsection{Running time and memory complexity}

Consider a sequence of $n$ consecutive calls to the $\textsc{Update}_{2+\eps}$ procedure, which must involve $\Theta(n)$ insertions and $O(n)$ deletions. 
In what follows, we analyze the performance of our algorithm relatively to these calls.
Let us focus on an arbitrary guess $\gamma \in \Gamma$. We first analyze the overall cost of the reclustering operations.

\begin{lemma}\label{lem:remain_vanishing}
If a point $x \in X$ is vanishing for $\gamma$ at time $t_0$, then $x$ remains vanishing for all $t$ such that $t_0 \leq t < t_{\rm del}(x)$. 
\end{lemma}

\begin{proposition}\label{prop:reclustering-bound}
    The total number of operations made by all calls of $\textsc{Reclustering}(\gamma)$ is $\BO{n \cdot k}$.
\end{proposition}

We now estimate the overall costs of insertions and deletions. For insertions, the following proposition
holds because of Proposition~\ref{prop:reclustering-bound}
and because, besides the call to $\textsc{Reclustering}(\gamma)$,
$\textsc{Insert}(p,\gamma)$ takes $\BO{k}$ time. 
\begin{proposition} \label{prop:insertion-bound}
    The total number of operations made by 
    all calls of $\textsc{Insert}(p,\gamma)$ is $O(n\cdot k)$. 
\end{proposition}

Consider now the deletions.  
In what follows, for time $t$ and $x \in X_t$, and for each $\gamma \in \Gamma$, let $c^{\gamma, t}_{ID}(x)$ be the index of the cluster in $\Clu^\gamma$ that $x$ is in at time $t$, and let $c^{\gamma,t}_{ID}(x) = k+1$ if $x\in U^\gamma$. 

\edit{To bound the total number of operations made by deletions we first show that once a point becomes vanishing the number of deleted centers with index smaller than the point's index is not greater than $k+1$. This information, together with Property (2), will later help upper bound the number of re-homing operations made at Line \ref{line:deletion-reclustering-for}.}

\begin{lemma}\label{lem:center-deletion-bound}
    Let $x\in X$ and $\gamma\in\Gamma$. Call $\mathcal{Q}^\gamma_x$ the set of centers $c$ such that 
    at time $t = \tdel(c)-1$ 
    it holds that $x$ is vanishing for $\gamma$ and $c^{\gamma, t}_{ID}(c)\le c^{\gamma, t}_{ID}(x)$.
    Then, $|\mathcal{Q}^\gamma_x|\le k+1$.
\end{lemma}

\begin{proposition} \label{prop:deletion-bound}
    The total number of operations made by 
    all calls of $\textsc{Delete}(p,\gamma)$ is $O(n\cdot k)$. 
\end{proposition}
\begin{proof}
    First, we bound the number of operations done by the re-assignment of cluster points in the loop at Line \ref{line:deletion-reclustering-for}.
    We begin by noticing that when $c_i$ is deleted, only persistent points contribute to the loop at Line~\ref{line:deletion-reclustering-for}.
    For a single such point $x \in \Clu_i, x \neq c_i$, Lines 6-13 perform $O\big(c^{\gamma, t}_{ID}(x) - c^{\gamma, t-1}_{ID}(x)\big)$ operations, i.e., proportional to the difference between the current index $i$ and the index of the cluster that it will end up in. If $c^{\gamma, t}_{ID}(x)$ was non-decreasing over time, then the maximum number of operations for $x$ across all timesteps would be $O(k)$. However this is not the case because of re-indexing due to center deletions.

    Given that $x$ must be a persistent point, we are only interested in events in which $c^{\gamma, t}_{ID}(x)$ decreases and $x$ remains a persistent point. 
    The only such case occurs during the relabeling process in a deletion that involves a center (Line \ref{line:del_relabeling}) with index $< c^{\gamma, t-1}_{ID}(x)$. 
    In this case $c^{\gamma, t}_{ID}(x)=c^{\gamma,t-1}_{ID}(x)-1$, so each center deletion adds 1 to the maximum number of re-assignment operations for each persistent point with index greater or equal to its own.

    Let $\delset:=\{c: c\in C^{\gamma, \tdel(c)-1}\}$ be the set of all deleted centers at any time. 
    We remark that, for $x \in X$, superscripts $\tdel(x)-1$ denote states right before the deletion of $x$.
    We get that the total number of re-assignment operations is:
    {\small
    \[  
        O\bigg(n\cdot k + \sum_{c\in\delset}\sum_{i\ge c^{\gamma, \tdel(c)-1}_{ID}(c)} |P^{\gamma, \tdel(c)-1}_i| \bigg)
        = O\bigg(n k + \sum_{c\in\delset} \Big(|U^{\gamma, \tdel(c)-1}|+\!\!\! \sum_{i\ge c^{\gamma,\tdel(c)-1}_{ID}(c)} \!\!\!\!\!\!|V^{\gamma,\tdel(c)-1}_i| \Big)\! \bigg),
    \]}
    where we used the \edit{negation of Property (\ref{eq:size_balance}), which holds after reclustering}. 
    If we fix $c\in\delset$ in the last summation and call $j:=c^{\gamma, \tdel(c)-1}_{ID}(c)$ we get that, for each $x\in U^{\gamma, \tdel(c)-1}\cup\bigcup V^{\gamma, \tdel(c)-1}_{i\ge j}$,
    $x$ is by definition vanishing and $c^{\gamma, \tdel(c)-1}_{ID}(c)\le c_{ID}^{\gamma, \tdel(c)-1}(x)$.
    Therefore we have:
    {\small
    \[ 
    \begin{split}
        \sum_{c\in\delset} \bigg(|U^{\gamma, \tdel(c)-1}|+\sum_{i\ge c^{\gamma, \tdel(c)-1}_{ID}(c)} |V^{\gamma, \tdel(c)-1}_i| \bigg)
        &\leq\sum_{x\in X} \Big|\Big\{c\in\delset: \substack{x \text{ is vanishing at time } \tdel(c)-1 \\\text{ and } c^{\gamma,\tdel(c)-1}_{ID}(c)\le c^{\gamma,\tdel(c)-1}_{ID}(x)}\Big\} \Big| \\
        &= \sum_{x\in X}|\mathcal{Q}^\gamma_x|\le n\cdot(k+1),
        \end{split}
    \]
    }
    
    where we used a double counting argument by rearranging the summation order, and then applied Lemma \ref{lem:center-deletion-bound}. 
    Hence the total number of operations made by the re-assignment of cluster points in the deleting function over all calls of $\textsc{Delete}(p,\gamma)$ is $O(n\cdot k)$.
    Second, we bound the number of operations made by the creation of a new cluster in the deleting function (Line \ref{line:del_create_cluster}). 
    When creating a new cluster we scan each element in $U$ for finding the one with greater $\tdel$ and then we do the same for finding which points are within distance $2\gamma$. Therefore, the total number of operations is $O(\sum_{c\in\delset}|U^{\gamma, \tdel(c)-1}|)$, which, as above, is $O(n\cdot k)$.

    Finally, all other operations in the deletion procedure take $O(k)$ time, which concludes the proof.
\end{proof}

Based on Propositions~\ref{prop:insertion-bound} and \ref{prop:deletion-bound}, we can prove the bound for update times. We also state the bound for query times, as well as the memory requirements of our algorithm.

\begin{theorem} \label{thm:upd-2} 
The total number of operations made by a sequence of $n$ consecutive
calls to $\textsc{Update}_{2+\eps}$ is $\BO{n ({\max_t\log|X_t|} + (\log\Delta / \eps) k)}$. 
\end{theorem}

\begin{theorem} \label{thm:query-2}
    The number of operations made by each $\textsc{Query}_{2+\eps}$ call is $\BO{(\log\Delta/\eps) + k}$.
\end{theorem}

\begin{theorem} \label{thm:mem-2}
    The maximum number of points maintained in memory when processing the pointset $X$ is $\BO{(\log\Delta/\eps) \max_t|X_t|}$.
\end{theorem}

Interestingly, our update times are tight, up to a factor $\BO{\log\Delta/\eps}$, due to the lower bound of $\Omega(|X|k)$ distance queries given by \citet[Theorem 54]{bateni2023optimal}, which holds even in the offline setting.

\subsection{Comparison with related work} 
\label{sec:comparison-2}

As we mentioned in the introduction, the dynamic setting with lifetimes bridges the sliding window and fully dynamic frameworks. We compare our algorithm to the existing fully dynamic ones, as they can be used even in the setting with lifetimes (while sliding windows ones would lose all guarantees).

Our algorithm mimics the fully-dynamic algorithm from \citet{chan2018fully}, yet by exploiting the known deletion times it lowers the amortized update times from $O\big( \frac{\log\Delta}{\eps}k^2\big)$ to $O\big(\frac{\log\Delta}{\eps}k\big)$, while having the same approximation ratio.
In fact, this allows us to have faster updates even compared to the more complex $(2+\eps)$-approximation algorithm by \citet{bateni2023optimal}, which has amortized $O\big((\log\Delta/\eps)\max_t\log|X_t|(k+\max_t\log|X_t|) \big)$ update time. 
Importantly, since our algorithm is deterministic, it works even against metric-adaptive adversaries, while the aforementioned two can only deal with oblivious ones.
Therefore, in settings where lifetimes are available, it would be wasteful to use a fully-dynamic algorithm, as they are more computationally expensive and with less guarantees, rather than our algorithm with lifetimes.
\edit{In fact, making use of lifetimes, our algorithm can \textit{sidestep} the metric-adaptive lower bound of \citet[Thm.~3]{bateni2023optimal}, which rules out fully dynamic $O(1)$-approximation algorithms with polylogarithmic update time for a range of values of $k$ against metric-adaptive adversaries. Recent work by \citet{grilnberger2026adaptive} obtains a randomized $O(1)$-approximation with $\widetilde{O}(k)$ update time against point-adaptive adversaries, which does not contradict the aforementioned lower bound, as the metric is fixed in advance.}

The recent paper by \citet{blank2026expiration} addresses a dynamic setting akin to the one with lifetimes, while focusing primarily on the memory requirements. Indeed, their algorithm requires $\BO{(\log\Delta/\eps) \cdot k^3}$ memory, as opposed to our requirements, which are linear in the active set size. However, their algorithm features a high $\BO{(\log\Delta/\eps) \cdot k^4}$ update time, and, more importantly, features only a $(6k+2+\eps)$ approximation ratio, which is much worse than our almost-tight $(2+\eps)$-approximation. %

\section{A space-efficient algorithm with constrained lifetimes}
\label{sec:6approx}
The algorithm developed in the previous section achieves an almost-tight approximation ratio, but requires memory which is linear in the size of the active set. In this section, we devise an algorithm which requires much smaller memory, constant with respect to the size of the active set, at the cost of a slightly worse approximation guarantee. 
\edit{As before, we assume that $|X_t|> k$ whenever a query is made. Queries when $|X_t| \leq k$ can be handled by maintaining the $k$ points with largest deletion times.}

Let $\eps>0$ be a user-defined approximation tolerance, and
let $\Gamma$ be the set of guesses defined by Equation~\ref{eq:gamma}, 
instantiated now with
 $\beta=\eps/6$.
At any time $t$, the algorithm maintains, for each $\gamma \in \Gamma$, a data structure, which is able to provide a solution of radius $6\gamma$ or a certificate that no solution of radius $\gamma$ exists. The data structure, \edit{which closely matches the one from \citet{cohen2016diameter} }, consists of the following sets:
\begin{itemize}
    \item $\Ag = (a^\gamma_1, \dots, a^\gamma_{\ell^\gamma}) \subseteq X_t$ a set of $\ell^\gamma \leq k+1$ \emph{attractors} such that $d(a^\gamma_i, a^\gamma_j) > 2\gamma, \forall i\neq j$. For $x\in X_t$, $a \in \Ag$, we say that $x$ is attracted by $a$ if $\tarr(x) \geq \tarr(a)$ and $d(x,a) \leq 2\gamma$.
    \item $R_{\rm act}^\gamma = (r^\gamma_1, \dots, r_{\ell^\gamma}^\gamma) \subseteq X_t$ a set of \emph{representatives}, where, for $1 \leq i \leq \ell^\gamma$, $r_i$ is the point with latest deletion time among the ones attracted by $a_i \in \Ag$, and it is denoted with $\repg(a_i) = r_i$. We use $R_{\rm orph}^\gamma \subseteq X_t$ to denote the set of representatives whose attractors are no longer in $A^\gamma$, and define $\Rg = R_{\rm act}^\gamma \cup R_{\rm orph}^\gamma$. 
\end{itemize}

\edit{Intuitively, for a given guess $\gamma$, the data structure maintains up to $k+1$ clusters of radius $2\gamma$, centered at the points $a^\gamma_i$'s. While in the data structure from Section 3 we maintained the entire clusters in memory, here we only maintain the point with longest lifetime for each cluster, called a representative.}

Observe that if $|A^\gamma| \geq k+1$, there are $k+1$ points at distance $>2\gamma$, which implies that $\gamma < r_k^*(X_t)$. 
Conversely, if $|A^\gamma| \leq k$, we will show that any point of $X_t$ is at distance at most $6\gamma$ from $R^\gamma$. 

\paragraph{Update procedure}
When a new point $p$ arrives, procedure $\textsc{Update}_{6+\eps}(p)$ is called
(see Algorithm~\ref{alg:update} for the pseudocode) to handle the insertion of $p$, as well as possible deletions of expired points. 
For every $\gamma \in \Gamma$, the data structures are updated as follows. First, all points expired before time $\tarr(p)$ are removed from $\Ag$ and $\Rg$. Then, if $p$ is at distance $>2\gamma$ from every attractor $a \in \Ag$, then $p$ is added to $\Ag$, and to $\Rg$ as representative of itself
(i.e., $p = \repg(p)$), and procedure \textsc{Cleanup}($p,\gamma$) is invoked to possibly evict unnecessary points from $\Ag$ and $\Rg$.
Otherwise, an attractor $a \in \Ag$ with
$d(a,p) \leq 2 \gamma$ and $\tdel(\repg(a)) < \tdel(p)$, if any, is arbitrarily chosen and $p$ becomes the new representative $\repg(a)$. If no such attractor is found, $p$ is discarded. 

Procedure \textsc{Cleanup}($p,\gamma$) (see Algorithm~\ref{alg:cleanup} for the pseudocode), works as follows.
If $|\Ag| = k+2$, then the attractor with minimum deletion time is removed from 
the set. After this removal, or if the size of $\Ag$ is $k+1$ at the start of 
the procedure, all  
representatives $q$ such that $\tdel(q) < t_{\rm min}$,
where $t_{\rm min} = \min_{a \in \Ag} \tdel(a)$ is the minimum deletion time of an attractor,
are removed from $\Rg$. 
These removals are justified by the following argument. If $|\Ag| = k+1$, then
these $k+1$ points, which are
at distance greater than  $2 \gamma$ from one another, 
provide a certificate that no solution of radius $\leq \gamma$ exists for
any $X_{t'}$ with $t \leq t' < t_{\rm min}$, hence $\gamma$ will not be a useful guess until $t_{\rm min}$, and there is no need to keep points 
that expire in this time interval.

\begin{figure}
\begin{minipage}[t]{0.585\textwidth}
\begin{algorithm2e}[H]
\caption{$\textsc{Update}_{6+\eps}(p)$}
\label{alg:update}
\SetAlgoVlined
\LinesNumbered
\ForEach{$\gamma \in \Gamma$}{
Remove points $x$ of $\Ag\!$ and $\Rg\!$ with $\tdel(x) \!\leq\! \tarr(p)\!\!\!\!$\\
\label{update-step}
    \Let{$E$}{$\{ a \in \Ag : d(a,p) \leq 2\gamma\}$}\label{line:update6eps-E} \\
    \If{$E = \emptyset$}{
      \Let{$\Ag$}{$\Ag \cup \{p\}$} \\
      \Let{$\repg(p)$}{$p$}; \, \Let{$\Rg$}{$\Rg \cup \{p\}$}\\
      \textsc{Cleanup}($p,\gamma$)\\
    }
    \Else{
    \Let{$E'$}{$\{a \in E \; : \; \tdel(\repg(a)) < \tdel(p)\}$} \\
    \If{$E' \neq \emptyset$}{
        \Let{$a$}{arbitrary attractor in $E'$} \\
        \Let{$\Rg$}{$\Rg \setminus \{\repg(a)\}$}\\
        \Let{$\repg(a)$}{$p$}; \, \Let{$\Rg$}{$\Rg \cup \{p\}$}\\
      }
    }
}
\end{algorithm2e}
\end{minipage}
\hfill
\begin{minipage}[t]{0.38\textwidth}
\begin{algorithm2e}[H]
\caption{$\textsc{Cleanup}(p, \gamma)$}
\label{alg:cleanup}
\SetAlgoVlined
\LinesNumbered

\If{$|\Ag| = k+2$} {
    \Let{$a_{\rm min}$}{$\argmin_{a \in \Ag} \tdel(a)$}\;
    Remove $a_{\rm min}$ from $\Ag$\;
\label{cleanup-step}
}
\If{$|\Ag| = k+1$}{
    \Let{$t_{\rm min}$}{$\min_{a \in \Ag} \tdel(a)$}\;
    Remove all points $q$ with $\tdel(q) < t_{\rm min}$ from $\Rg$
}
\end{algorithm2e}
\end{minipage}
\end{figure}

\textbf{Query procedure} $\ $
At any time $t$, to obtain a solution to $k$-center for the active set
$X_t$, procedure $\textsc{Query}_{6+\eps}(t)$ is invoked (see Algorithm~\ref{alg:query-6} in the appendix for the pseudocode). 
First, for each $\gamma \in \Gamma$, all points $p$ with $\tdel(p) \leq t$  are removed from $\Ag$ and $\Rg$. Then, the procedure searches
the minimum guess $\gamma \in \Gamma$ such that $|\Ag| \leq k$
and a $2\gamma$-covering $C^\gamma$ of $\Rg$ of size at most $k$ can be found using the simple greedy strategy by \citet{hochbaum1986unified}, and this set $C^\gamma$ is returned as the solution.

\subsection{Approximation} \label{sec-analysis}

In what follows, we denote with superscript $t$ the states of data structures either after the call of $\textsc{Update}_{6+\eps}(p)$, with $\tarr(p) = t$, or after the execution of the first for loop in $\textsc{Query}_{6+\eps}(t)$.
To prove the approximation ratio, we first prove that the algorithm maintains the following invariant.

\begin{lemma} \label{lem:invariants}
After any execution of $\textsc{Update}_{6+\eps}(p)$, the following
properties hold for each $\gamma \in \Gamma$, and $t = \tarr(p):$
\begin{enumerate}
\item\label{lemprop1} 
If $|\Agt[t]| \le k$, then $\forall x\in X_t$ we have 
$
d(x, \Rgt[t]) \le 4\gamma.
$
\item\label{lemprop2} 
If $|\Agt[t]| = k+1$, then $\forall x\in X_t$ 
such that $\tdel(x) \geq \min_{a \in \Agt[t]} \tdel(a)$, we have
$
d(x, \Rgt[t]) \le 4\gamma.
$
\end{enumerate}
\end{lemma}

\edit{At a high level, the invariant states that if there are at most $k$ attractors, then any active point is at distance at most $4\gamma$ from a representative. If however there are $k+1$ attractors, which implies that $\gamma$ is too small to be the correct guess, then only the points with deletion time later than the earliest deletion time of an attractor are guaranteed to be at distance at most $4\gamma$ from a representative. }

Based on the invariants, we can prove the approximation ratio of the solutions returned by $\textsc{Query}_{6+\eps}$. 
\begin{theorem} 
\label{thm:approx-6}
For any time $t$, the solution returned by $\textsc{Query}_{6+\eps}(t)$ is a $(6+\epsilon)$-approximation for the $k$-center problem on the active set $X_t$.
\end{theorem}

\subsection{Running time and memory complexity}
We now analyze the number of points maintained by the data structure employed by our algorithm, which directly influences the time and memory requirements. The analysis is made parametric with respect to a
specific notion of ordering of the input dataset $X$, which is defined below and it is a crucial novel ingredient introduced by our paper.
\begin{definition}  
\label{def:h-bounded}
For an integer parameter $H \geq 0$, we say that the set $X$ is \emph{$H$-ordered} if for every pair of points $p,q \in X$ such that $\tarr(p)<\tarr(q)$ and at least $H$ points arrive between $p$ and $q$, i.e., $\left|\{ x \in X \colon \ \tarr(p) < \tarr(x) < \tarr(q) 
\}\right| \geq H$,
we have that $\tdel(p) < \tdel(q)$.
\end{definition}  
Intuitively, the above notion provides a quantitative characterization of how
arbitrary the points' arrival and departure times are, and a low value of $H$ means that points are deleted roughly in the order in which they arrive. For example, in the sliding-window setting, or in general when points are deleted in the order they arrive, the pointsets are 0-ordered. 
\edit{More generally, any first-in-first-out (FIFO) policy is $0$-ordered, %
and hence this class strictly contains the sliding-window setting.}

\edit{The assumption of $H$-ordered sets is realistic. A concrete example is provided by the aforementioned TLS certificates recorded in Certificate Transparency logs \cite{cooper2008internet, laurie2021certificate}. These certificates are valid for 160 hours since issuance and have therefore a set expiration time, and they are usually recorded in a Certificate Transparency log shortly after issuance. Taking the logging time as arrival, we then have that their deletion order closely follows their arrival order, with a few inversions due to uneven logging delay, which yields a stream with low but nonzero $H$.}
As another example, consider a stream with batch processing (e.g., buffers in TCP packet routers, or intermediate storage for industrial manufacturing). Items arrive in sequence and are grouped into batches of size $H$, with each batch being processed as a unit. While the deletion order of items within a batch may change, batches themselves are served in arrival order, resulting in an $H$-ordered set.

Importantly, the algorithm does not require the knowledge of $H$, which is only used by the analysis \edit{to bound the size of the data structure}.
\edit{Indeed, intuitively, the attractors and their active representatives can be at most $O(k)$, so the only points that could grow up to $|X_t|$ are orphaned representatives. Under H-ordered sets, however, their number is bounded by $O(k+H)$. We then have the following.}

\begin{lemma} \label{lem:space-bound-6}
Suppose that $X$ is $H$-ordered, for some $H \geq 0$. Then, for any guess $\gamma \in \Gamma$ and time $t$, 
\[
|\Agt[t]|+|\Rgt[t]| = \BO{\min\{(k+H),|X_t|\}}.
\]
\end{lemma}

The next theorems, based on Lemma~\ref{lem:space-bound-6}, state upper bounds for the update and query running times, and for the memory complexity. Let ${\rm pr}(t)$ be the time of the last update/query operation before $t$.

\begin{theorem} \label{thm:upd-6}
$\textsc{Update}_{6+\eps}(p)$ requires $\BO{(\log\Delta/\eps) \min\{(k\!+\!H),|X_{{\rm pr}(\tarr(p))}|\}}$ operations.
\end{theorem}

\begin{theorem} \label{thm:query-6}
$\textsc{Query}_{6+\eps}(t)$ requires $\BO{(\log\Delta/\eps) k \cdot \min\{(k+H),|X_{{\rm pr}(t)}|\}}$ operations.
\end{theorem}

\begin{theorem} \label{thm:mem-6}
The maximum number of points maintained
in memory when processing the pointset $X$ is 
$\BO{(\log\Delta/\eps) \min\{(k+H), \max_t|X_t|\}}$. 
\end{theorem}

Note that, for $H \in O(k)$, update times and the memory complexity are bounded by 
$\BO{(\log\Delta/\eps)k}$.

\subsection{Comparison with related work} \label{sec:comparison-6}

The memory and time complexities of the algorithm presented in this section depend on the tameness of the update sequences. Importantly, the memory usage never exceeds the one of the algorithm presented in Section~\ref{sec:2approx}, and,
for small $H$ (e.g., $H \in O(k)$), it is significantly smaller memory, while attaining comparable update time. Moreover, unlike any existing fully-dynamic algorithm in general metrics, the bound on the operations is \textit{worst-case} rather than amortized.
These gains, however, come at the cost of a (slightly) worse approximation ratio.

Our setting can be viewed as a generalization of the sliding window one. Indeed, our algorithm extends the sliding window one of \citet{cohen2016diameter} by allowing arbitrary deletions. Despite the increased generality, we match the same memory and time complexities when $H=0$. 
As discussed in Section~\ref{sec:comparison-2}, \citet{blank2026expiration} present an algorithm for $k$-center in a setting akin the one with lifetimes, which stores $\BO{(\log\Delta/\eps) \, k^3}$ points regardless of $H$-ordering and is more memory-efficient than ours for large $H$ (i.e., $H \in \Omega(k^3)$). 
However, their algorithm has expensive update times, and, more importantly, a worse non-constant approximation ratio, growing linearly with $k$.

\section{Conclusions}
\label{sec:concl}

In this paper, we explored the dynamic setting \textit{with lifetimes}, devising constant-approximation $k$-center clustering algorithms. The new setting represents a trade-off between the fully dynamic and the sliding window settings, and offers several opportunities for future work. It is interesting to study variants of $k$-center problem (e.g., dealing with outliers \citep{biabani2024improved} or minimizing recourse \citep{lkacki2024fully}), and to apply the new setting to other problems, including clustering with different objectives \citep{bhattacharya2024fully}, graph problems \citep{cruciani2024dynamic}, and machine unlearning problems \citep{bourtoule2021machine}. Moreover, a compelling goal is to consider a more general setting where only approximate estimates of deletion times are available, linking the framework to the broader field of \textit{algorithms with predictions} \cite{mitzenmacher2022algorithms, LiuV24}. 
\edit{Finally, whether the approximation-memory trade-off achieved by our space-efficient algorithm is optimal remains open, since no lower bound specific to the setting with lifetimes is currently known.}

\newpage

\bibliographystyle{plainnat} 
\bibliography{references}

\newpage

\appendix

{\LARGE \textbf{Appendix}}

\section{Comparison with settings akin to the one \textit{with lifetimes}} \label{app:akin}

Dynamic settings akin to the one \textit{with lifetimes} introduced here have recently been proposed in the literature. 

\citet{PengR22} address the \textit{deletions-look-ahead} dynamic model, where the order of deletions of elements of the active set is known to the algorithm, like in the dynamic setting with lifetimes. They show that any incremental algorithm with worst-case update time $T_u$ and query time $T_q$ can be turned into an algorithm that can support both insertions and deletions in the deletions-look-ahead model with $O(T_u \log|X|)$ amortized update time, $O(T_q)$ query time, and a $\Theta(\log|X|)$ memory overhead. Crucially, the total number of elements that are inserted or deleted from the active set, i.e., $|X|$, can be exponentially bigger than the number of active elements at any point in time, i.e., $\max_t |X_t|$. Because of this, the proposed technique would fail the requirement of the dynamic setting with lifetimes requiring update times sublinear in $|X_t|$. 

The recent paper of \cite{blank2026expiration} introduces the \textit{general expiration streaming model}, which generalizes the sliding window model by allowing for arbitrary deletion orders, yet elements are endowed with their deletion time, like in the dynamic setting with lifetimes. Since this model stems from the streaming one, the focus seems to be primarily on sublinear memory usage. For a comparison with their algorithm, see Section~\ref{sec:comparison-6}.

Finally, several settings have been introduced to inform dynamic algorithms with predicted or uncertain (rather than exact) deletion times of elements, inserting these works in the broader \textit{algorithms with predictions} framework \citep{mitzenmacher2022algorithms}.
Among these, \citet{LiuV24} utilizes the \textit{predicted-update dynamic model}, extending the setting of \citet{PengR22} to allow for predicted deletion times, and show that also under this setting some incremental or offline algorithms can be turned into fully-dynamic ones, with $O({\rm polylog}|X|)$ overhead. 
Moreover, both \citet{brand2024dynamic} and \citet{HenzingerSSY24} endow fully dynamic algorithms with predicted deletion times of points, in a setting akin to the predicted-update dynamic model of \citet{LiuV24}, but focus on graph problems.

\section{Limitations and future work} \label{sec:limitations}

In this paper, we developed two algorithms for k-center in the newly introduced dynamic setting \textit{with lifetimes}. 
While the introduction of this setting and the new algorithms offer several advantages over algorithms for the fully dynamic or the sliding window setting (in terms of efficiency and generality, respectively), we envision that future work can address several limitations of the present work. 

First, developing lower bounds on the trade-offs between approximation guarantees and time complexities for k-center in the dynamic setting with lifetimes, and understanding how they differ from the ones in the fully dynamic setting (e.g., the ones presented in \citet{bateni2023optimal}), would be a major step forward for the field. Additionally, it would be interesting to understand how these lower bounds change under the requirement of sublinear working memory (i.e., making the dynamic setting with lifetimes analogous to the general expiration streaming model \citep{blank2026expiration}), and how they relate to the ones in the sliding window model \citep{cohen2016diameter}.

Second, it would be interesting to relax the assumptions of the algorithms we developed, for example by removing the need for the knowledge of the aspect ratio $\Delta$ of the metric space, or by allowing for inexact deletion times, which would bridge the dynamic setting with lifetimes to the algorithm with predictions literature \citep{mitzenmacher2022algorithms, brand2024dynamic, HenzingerSSY24}. 

Moreover, as for the fully dynamic setting, the dynamic setting with lifetimes allows for a definition of \textit{recourse}, defined as the number of changes (i.e., $|C_1 \triangle C_2|$) made to
the cluster centers after each point insertion or deletion \citep{lattanzi2017consistent}. Both the algorithms we present here have trivial $O(k)$ recourse, and developing algorithms with $o(k)$ recourse is an interesting research direction. In the fully dynamic setting, \citet{forster2025dynamic} developed an algorithm that yields a $50$-approximation with optimal recourse of $2$ and $\widetilde{O}({\rm poly}(|X|))$ update time, while \citet{bhattacharyaalmost} developed an algorithm that yields a $20$-approximation with $O(1)$ recourse and $\widetilde{O}(k)$ update time. Understanding how these trade-offs can be improved in the dynamic setting with lifetimes would be a major development. 

Finally, several extensions to the k-center problem can be tackled in the dynamic setting with lifetimes. For example, the algorithms could be extended to deal with variants of the problem, including dealing with outliers \citep{biabani2024improved}, diversity maximization \citep{wang2023fair}, and fairness constraints \citep{ceccarello2025fair}.
Moreover, techniques for low-dimensional metric spaces could be used to improve the approximation guarantees under memory constraints \citep{pellizzoni2022adaptive} or to obtain $O(1)$ update times \citep{gan2024fully, pellizzoni2023fully}.

\section{Additional pseudocodes}

We provide here, for the sake of completeness, the pseudocodes for the query procedures. 
In particular, Algorithm~\ref{alg:query-2} reports the query procedure for the algorithm presented in Section~\ref{sec:2approx}.

\begin{algorithm2e}[h!] 
\SetAlgoLined
$\gamma^* \leftarrow \min \{\gamma \in \Gamma \; : \; U^\gamma = \emptyset\}$\;
\Return $C^{\gamma^*}$
\caption{$\textsc{Query}_{2+\eps}(t)$} \label{alg:query-2}
\end{algorithm2e}

Moreover, Algorithm~\ref{alg:query-6} reports the query procedure for the algorithm presented in Section~\ref{sec:6approx}.

\begin{algorithm2e}[h!]
\SetAlgoVlined
\LinesNumbered
\ForEach{$\gamma \in \Gamma$}{
Remove points $x$ of $\Ag$ and $\Rg$ with $\tdel(x) \!\leq\! t\!$ \\
}
\For{increasing values of $\gamma \in \Gamma$ such that $|\Ag| \le k$}{
    \Let{$C$}{$\emptyset$}\;
    \ForEach{$q\in \Rg$}{
        \lIf{$C = \emptyset \text{ or }  d(q, C) > 2\gamma$}{
            \Let{$C$}{$C \cup \{q\}$}
        }
        \lIf{$|C| > k$}{go to next $\gamma$}
    }
    \Return $C$
        
}
\caption{$\textsc{Query}_{6+\eps}(t)$}
\label{alg:query-6}
\end{algorithm2e}

\section{Omitted Proofs}

\subsection{Proofs for Section~\ref{sec:2approx}}

\begin{proof}[Proof of Remark~\ref{remark:sequence}]
We describe an adversarial point set for a given guess $\gamma \in \Gamma$. At time $t=1$, we insert a point $p_1$ with lifetime such that $\tdel(p_1) = \tarr(p_1)+n+1$.

At time $t=2, \dots, n$, we insert $n-1$ points $p_2, \dots, p_{n}$, each with lifetime such that $\tdel(p) = \tarr(p) + 2n-1$, and such that each $p_i$ with $i>1$ is at distance $\nicefrac{3}{2}\gamma$ from $p_1$ (and, e.g., at distances $\nicefrac{\gamma}{10}$ between each other). Clearly, all these points are assigned to cluster $\Clu_1$, with the first point $p_1$ being the center.
Then, at time $t=n+1, \dots, 2n$, we insert $n$ additional points $p_{n+1}, \dots, p_{2n}$, each with lifetime such that $\tdel(p_i) = \tarr(p_i) + 2$ and such that each such $p_i$ is at distance $\nicefrac{5}{2}\gamma$ from $p_1$, at distance $\nicefrac{5}{2}\gamma$ from the other points $p_j$ with $n < j \leq 2n$, and at distance $\nicefrac{3}{2}\gamma$ from points $p_2, \dots, p_n$. 

The first $n$ calls of $\textsc{Update}_{2+\eps}$ involve no deletions, so they take $O(k)$ time each. The call at time $n+1$ inserts $p_{n+1}$ as a new center, since it is at distance $> 2\gamma$ from $p_1$, in $O(k)$ time. 
However, the call at time $n+2$ first involves the deletion of $p_1$ and moving $p_2, \dots, p_{n}$ to the cluster with center $p_{n+1}$, and then inserts $p_{n+2}$ as a new center. This takes $\Theta(n)$ time. 
Each call at time $t = n+3, \dots, 2n$ involves the deletion of point $p_{t - 2}$, moving $p_2, \dots, p_{n}$ to the cluster with center $p_{t-1}$ and the insertion of $p_{t}$ as a new center, which takes $\Theta(n)$. 
Therefore, the algorithm performs a total of $\Theta(n^2)$ operations.
\end{proof}

\begin{proof}[Proof of Lemma~\ref{lem:invariant}]
    At step $t=0$, the property holds by vacuity. 
    Suppose inductively that at step $t \geq 0$, the invariant holds for all points in $X_t$. We recall that the superscript (e.g. $U^{\gamma, t}$) denotes the state of a set \textit{after} the $t$-th insert or delete operation. When the superscript is missing, it denotes a state during the execution of the operations. At time step $t+1$:

    \begin{itemize}
        \item  If $\textsc{Insert}(p, \gamma)$ is called, then all points in $X_t$ are left untouched. Moreover, the invariant for $p$ holds by construction. 
        \item If $\textsc{Delete}(p, \gamma)$ is called and $p$ is not a center, then all points in $X_t \setminus \{  p\}$ are left untouched. 
    Let instead $p$ be a center $c^{\gamma, t}_i$.
    Then, for each $x \not\in \Clu^{\gamma, t}_i \cup U^{\gamma, t}$, note that the center and cluster index renaming does not affect the invariant. 
    For each $x \in \Clu_i^{\gamma, t}$, either $x$ is inserted in $\Clu_h$ with $h = \argmin_{j=i+1, \dots \ell}\{ c_j \in C \colon \ d(x, c_j) \leq 2\gamma \}$ (possibly by creating new clusters), or it is inserted in $U$ if no center at distance $\le2\gamma$ is found.
    In the first case, note that $\min_{j = 1, \dots, i-1 } d(x, c_j) > 2 \gamma$, otherwise by the inductive hypothesis $x$ would not belong to $\Clu_i$. Therefore, we have that $h = \argmin_{j=1, \dots \ell}\{ c_j \in C \colon \ d(x, c_j) \leq 2\gamma \}$, thus fulfilling the invariant. 
    
    Finally, let $U \neq \emptyset$ and consider a $x \in U$, either because it belonged in $U^{\gamma, t}$ or because it was inserted in $U$ on line~\ref{line:delete-insU}. In this case, one point $u$ will be selected as a new center, and thus $|C| = k$.
    We have that if $d(x, C^{\gamma, t+1}) > 2\gamma$, then we have $d(x, u) > 2\gamma$ and therefore $x$ remains in $U$, satisfying the first case of the invariant. Otherwise, since $x\in U$ implies that $d(x, c_i) > 2\gamma \ \forall i=1, \dots, \ell-1$, we have that $d(x, u) \leq 2\gamma$, and thus $x$ is inserted in $\Clu_\ell$, thus satisfying the second case of the invariant. 
    \end{itemize}
    Finally, after a call to \textsc{Insert} or \textsc{Delete}, the $\textsc{Reclustering}(\gamma)$ procedure might be called. 
    In that case, we note that each $x \in \bigcup_{j=1}^{i-1} \Clu_j^{\gamma, t}$ is left untouched, and the invariant is still valid. 
    On the other hand, if $x \in U'$ after line \ref{line:reclustering_cluster_joining}, then it could not be inserted in any of the clusters $\Clu_j$ for $j = 1, \dots, i-1$, by the inductive hypothesis.
    Moreover, by construction it will be inserted in the cluster $\Clu_h$ with $h = \argmin_{j=i, \dots \ell}\{ c_j \in C \colon \ d(x, c_j) \leq 2\gamma \}$, as a point can be inserted in a cluster only if it is not inserted in a cluster with a lower index. At the end of the loop (i.e., $|C|=k$), any point $x$ such that $d(x, C) > 2\gamma$ is left in $U$. 

    We therefore have the claim by induction.
\end{proof}

\begin{lemma}\label{lem:reclustering_makes_all_vanishing}
After a call of $\textsc{Reclustering}(\gamma)$ every point that was reclustered, i.e., every point belonging to the set $U'$
in $\bigcup_{j=i}^{\ell^\gamma} \Clu^\gamma_j \ \cup U^\gamma$ before the reclustering, becomes vanishing for $\gamma$.
\end{lemma}
\begin{proof}[Proof of Lemma~\ref{lem:reclustering_makes_all_vanishing}]
    Since at each step $j$ the new center $c_j$ gets chosen from all remaining points as the one with maximal deletion time, any point $x$ that gets inserted into $\Clu_j$ must have deletion time $t_{\rm del}(x) < t_{\rm del}(c_j)$, and is therefore vanishing. If a point is left in $U$, then it is vanishing by definition.
\end{proof}

\begin{proof} [Proof of Lemma~\ref{lem:remain_vanishing}]
    We will show that the vanishing property (for a fixed guess $\gamma$) is invariant under the call of $\textsc{Insert}$, $\textsc{Delete}$ or $\textsc{Reclustering}$ algorithms, unless $t=\tdel(x)$ and the point is thus removed. 
    The property holds at $t=0$ by vacuity and the proof follows by induction.

    Fix some $x\in X$ which is vanishing at time $t$.
    \begin{itemize}
        \item  \textbf{Insertion}: If $x$ is vanishing then it is already inserted, so any insertion operation will regard a new point and cannot modify either $x$ or its center. Then, $x$ will remain vanishing after any insertion operation.

        \item \textbf{Deletion}: Suppose that we call $\textsc{Delete}(p, \gamma)$. If $p=x$ the current time is $\tdel(x)$ and the proposition holds. 
        Furthermore, if $p\notin C^\gamma$ (i.e. $p$ is not a center) its deletion will only alter $p$ itself and no other point. 
        Therefore, we can assume $p\neq x$ and $\exists i$ s.t. $p=c_i$.
        We then have 3 cases:
        \begin{enumerate}
            \item $x \in \Clu_i$. This would mean that the center of $x$ has a deletion time smaller than that of $x$, which is in contradiction to the fact that $x$ is vanishing.
            \item $x \in \Clu_j$, for some $j \neq i$. In the case that $p=c_i$ the deletion algorithm affects clusters with index different from $i$ by possibly adding more points to them (Line \ref{line:del_insert_in_C_j}) and by relabeling their indexes (Line \ref{line:del_relabeling}). Therefore, no changes are made to $x$ or to its center $c_j$, and thus $x$ remains vanishing.
            \item $x \in U$. In this case, $x$ remains in $U$ or is inserted in $\Clu_l$ (Line \ref{line:del_insert_in_C_l}). Either way, it remains vanishing, since we picked $c_l = u$ s.t. $\tdel(u)$ is maximal among all points in $U$.
        \end{enumerate} 

        \item \textbf{Reclustering}: Suppose that we call $\textsc{Reclustering}(i,\gamma)$. If $\exists j<i$ such that $x\in \Clu_j$, then the reclustering operation does not affect $\Clu_j$ in any way, meaning that $x$ remains vanishing. The other case is covered by Lemma \ref{lem:reclustering_makes_all_vanishing}, and thus $x$ remains vanishing.
    \end{itemize}
\end{proof}

\begin{proof} [Proof of Proposition~\ref{prop:reclustering-bound}]
    The main for loop of Algorithm \ref{alg:reclustering} is repeated $O(k)$ times. In each iteration we need to scan $U'$ linearly to first find the element with maximal $\tdel$ and then to check whether some other point in $U'$ is close to it. So the main loop is executed in $O(k\cdot|U'|)$. The other parts of the algorithm can be executed in $O(|U'|)$, for a total running time of $O(k\cdot|U'|)$. Then:

\[
    k|U'|=
    k\left|U\cup\bigcup_{j=i}^k \Clu_j\right| \le 
    k\left(|U| + \sum_{j=i}^k \left(|P_j| + |V_j|\right)\right) \le 
    2k\sum_{j=i}^k |P_j|,
\]
    where the last inequality is guaranteed by the fact that Property (\ref{eq:size_balance}) holds for $i$ when calling the reclustering.

    This means that the number of operations is $O\left(k\cdot\sum_{j=i}^l|P_j|\right)$ and thus bounded by the number of persistent points involved in the reclustering.  
    Given that by Lemma \ref{lem:reclustering_makes_all_vanishing} all persistent points in $P_{j\ge i}$ will become vanishing and that by Lemma~\ref{lem:remain_vanishing} every vanishing point remains vanishing, the number of persistent points in all reclusterings is bounded by $n$ and thus the total number of operations made by all calls of $\textsc{Reclustering}(i,\gamma)$ is $O(n\cdot k)$
\end{proof}

\begin{proof}[Proof of Lemma~\ref{lem:center-deletion-bound}]
    We will show that for $c\neq x$ to satisfy the two properties implies the fact that when $c$ was last picked as a center $x$ was either not inserted or persistent. 
    Suppose by contradiction that, when $c$ was last picked as a center, $x$ was vanishing. We have 3 possibilities:
    \begin{enumerate}
        \item $c$ became a center for the last time at time $t'$ during an insertion operation (Line \ref{line:insert-center-addition}) or during a re-assignment of persistent points after a deletion (Line \ref{line:delete-center-addition}). In both of these cases $c^{\gamma, t'}_{ID}(c)> c^{\gamma, t'}_{ID}(x)$ as $c$ becomes the center of the cluster with the highest index. Given that we are assuming that $t'$ is the last time that $c$ was chosen as a center, there cannot be any reclustering involving $c$ after time $t'$. Furthermore, if no reclustering is made, clusters stay in their relative order and point $x$ will remain inside its cluster until its deletion. So the fact that $c^{\gamma, t'}_{ID}(c)\ge c^{\gamma, t'}_{ID}(x)$ is in contradiction with $c^{\gamma, t}_{ID}(c)\le c^{\gamma, t}_{ID}(x)$.
        \item $c$ became a center for the last time after a deletion. If $x\notin U$ the analysis is akin to the one in the first case. If $x\in U$ we know that $\tdel(x)<\tdel(c)$, thus at time $t$ $x$ will already be deleted and therefore cannot be vanishing.
        \item $c$ became a center for the last time during the execution of $\textsc{Reclustering}(i,\gamma)$. If $x\in \Clu_j$ s.t. $j<i$ the analysis is akin to the one in the first case. Let instead $x\in U'$ (same $U'$ as in Algorithm \ref{alg:reclustering}). If $x$ ends up in a cluster with smaller index than the one of $c$, then we break the assumption that $c^{\gamma, t}_{ID}(c)\le c^{\gamma, t}_{ID}(x)$. If instead it ends up in a cluster with index greater or equal to the one of $c$, then we must have $\tdel(x) <  \tdel(c)$, and thus $x$ cannot be vanishing at $t = \tdel(c)-1$ because it has been already deleted.
    \end{enumerate}

    Given that by Lemma~\ref{lem:remain_vanishing} a vanishing point remains vanishing, calling $t'$ the time at which $x$ becomes vanishing, every center $c$ which satisfies the two properties of the Lemma must be picked at time $<t'$ and carried on up to a time $\ge t'$. Given that at most $k$ centers can coexist simultaneously, the maximum number of centers $c\neq x$ satisfying the properties is $k$.  Accounting for the possibility that $c=x$, we obtain at most $k+1$ centers.
\end{proof}

\begin{proof}[Proof of Theorem~\ref{thm:upd-2}]
By implementing the priority queue $Q$ as a heap, 
each of the $\Theta(n)$ insertions and $O(n)$ remove-min operations in $Q$ costs $\BO{\log|X_t|}$, so their overall cost is $\BO{n \cdot {\max_t \log|X_t|}}$. By Propositions~\ref{prop:insertion-bound} and \ref{prop:deletion-bound}, for each $\gamma \in \Gamma$, we have that the overall costs of all calls $\textsc{Insert}(p,\gamma)$ and $\textsc{Delete}(p,\gamma)$ is  $\BO{nk}$. Since  $|\Gamma| \in \BO{\log_{1+\eps}\Delta}=\BO{\log\Delta/\eps}$, the theorem follows.
\end{proof}

\begin{proof}[Proof of Theorem~\ref{thm:query-2}]
    Note that we iterate through the guesses $\gamma \in \Gamma$, and for each guess we check in constant time the size of $C^\gamma$. This takes $O(|\Gamma|) = O(\log\Delta/\eps)$ time. Finally, returning the solution takes $O(|C^{\gamma^*}|) = O(k)$ time.
\end{proof}

\begin{proof}[Proof of Theorem~\ref{thm:mem-2}]
    Note that, by construction, after the call of $\textsc{Update}(p, t)$, only points of $X_t$ are maintained in memory. Moreover, for a fixed guess $\gamma$, any such point $p \in X_t$ can belong, by construction, to only one of the sets $\Clu^\gamma_i, i =1, \dots, \ell^\gamma$ or $U^\gamma$. We therefore have that $\sum_{i=1}^{\ell^\gamma} |\Clu^\gamma_i| + |U^\gamma| \leq |X_t|$. Moreover, clearly $|C^\gamma| \leq |X_t|$ as the centers are all distinct. 
    Since we have $|\Gamma| \in O(\log\Delta/\eps)$, the data structure maintains at most $O((\log\Delta/\eps) |X_t|)$ points. Moreover, the priority queue $Q$ contains each point of the active set exactly once, for a total of $|X_t|$ points. 
    Finally, until the following call of $\textsc{Update}$, the number of points stored in the data structures can only decrease (i.e., if $\textsc{Query}$ calls $\textsc{Delete}$), thus concluding the proof.
\end{proof}

\subsection{Proofs for Section~\ref{sec:6approx}}

\begin{proof}[Proof of Lemma~\ref{lem:invariants}]
Let us focus on an arbitrary guess $\gamma$. 
We say that a point $q \in X_t$ is
\emph{relevant for $\gamma$ and $t$} if $|\Agt[t]| \leq k$
or $|\Agt[t]| > k$ and $\tdel(q) \ge \min_{a\in \Agt[t]} \tdel(a)$.
Clearly, the proof concerns only relevant points. 
Also, for every point $q \in X_t$, we let $\psi^{\gamma}(q)$ denote the attractor that $q$ is attracted to. If $q$ is selected as a representative by $\textsc{Update}_{6+\eps}(q)$, we let $\psi^{\gamma}(q)$ be such that $\repg(\psi^{\gamma}(q)) = q$. Otherwise, we let $\psi^{\gamma}(q)$ be any attractor belonging to the set $E$ during the execution of $\textsc{Update}_{6+\eps}(q)$. 
Clearly, $d(q,\psi^{\gamma}(q)) \leq 2 \gamma$.
We will  argue  that for every point $q \in X_t$ which is relevant for $\gamma$ and $t$, there exists a point $r \in \Rgt[t]$ 
such that $\psi^{\gamma}(r) = \psi^{\gamma}(q)$. 

If $q$ is the point arrived at time $t$ (i.e., $q=p$ and $\tarr(q)=t$) then we set $r = rep^{\gamma, t}({\psi^{\gamma}(q)})$, which is possibly $q$ itself. 
If $q$ arrived at an earlier time,  it is immediate to argue that since $q$ is relevant at time $t$, it has also been relevant at all times $t'$, with $\tarr(q)\leq t' <t$, when a point arrived or a query was invoked. 
Set $r_0=rep^{\gamma,\tarr(q)}({\psi^{\gamma}(q)})$.
We have that $\psi^{\gamma}(r_0)=\psi^{\gamma}(q)$. Also, it holds that $\tdel(r_0)\geq \tdel(q)$, therefore $r_0$ is also relevant at all times $t'$, with $\tarr(q)\leq t' <t$. Being relevant, $r_0$  has  never  been removed from $\Rgt[t']$ during some invocation of \textsc{Cleanup}, which only removes non-relevant points. Consequently, since  $r_0$ is in $\Rgt[\tarr(q)]$, either $r_0$ is still in  $\Rgt[t]$ (and we set $r=r_0$), or it has been  expunged from some set $\Rgt[\tarr(r_1)]$,  due to the arrival of some longer-lived representative $r_1$
at time $\tarr(r_1)>\tarr(q)$.  Observe that
$\psi^{\gamma}(r_1)= \psi^{\gamma}(r_0)$, and that $r_1$ is also relevant at time $t$. If $r_1\in \Rgt[t]$, we set $r=r_1$, otherwise, we apply the same reasoning to $r_1$. Iterating the argument, we determine a sequence of relevant points $r_0, r_1, \ldots, r_h$, with $\tarr(r_0)\leq \tarr(q)< \tarr(r_1)\ldots <\tarr(r_h)$ and 
$\psi^\gamma(q)=\psi^\gamma(r_0) = \ldots = \psi^\gamma(r_h)$, with 
$r_h\in \Rgt[t]$. We then set $r=r_h$. 
The lemma follows immediately since $\psi^{\gamma}(r) = \psi^{\gamma}(q)$ implies
$d(q,r) \leq 4 \gamma$.
\end{proof}

\begin{proof}[Proof of Theorem~\ref{thm:approx-6}]
Let $\gamma = (1+\beta)^i \in \Gamma$ be the guess such that the returned solution $C$ is computed from $\Rgt[t]$, for some 
$\lfloor \log_{1+\beta} \dmin \rfloor \leq i \leq 
\lceil \log_{1+\beta} \dmax \rceil$. By construction, 
$|\Agt[t]| \leq k$, and, for each $p \in \Rgt[t]$, 
$d(p,C) \leq 2 \gamma$. Using Lemma~\ref{lem:invariants} and the triangle inequality, we have that
for each $p \in X_t$, $d(p,C) \leq 6 \gamma$. Now, if $i = \lfloor \log_{1+\beta} \dmin \rfloor$ then 
$\gamma \leq r_k^*(X_t)$ and the theorem
holds. Instead, if $i > \lfloor \log_{1+\beta} \dmin \rfloor$, we have, by construction, that for $\gamma'=(1+\beta)^{i-1}=\gamma/(1+\beta)$ either $|A^{\gamma',t}|=k+1$
or $k+1$ points of $R^{\gamma',t}$ have been found at distance $> 2\gamma'$
from one another. From what we observed earlier, in either case we have $\gamma' < r_k^*(X_t)$, hence 
$\gamma < (1+\beta)r_k^*(X_t)$. Thus,  for each $p \in X_t$,
$d(p,C) \leq 6 \gamma < 6 (1+\beta)r_k^*(X_t) \leq (6+\eps) r_k^*(X_t)$,
and the theorem follows.
\end{proof}

\begin{proof}[Proof of Lemma~\ref{lem:space-bound-6}]
Let us fix arbitrary $\gamma$ and $t$, and observe that, straightforwardly, $|\Agt[t]|+|\Rgt[t]| \leq 2|X_t|$.
We will now show that $|\Agt[t]|+|\Rgt[t]| = \BO{k+H}$.
By construction, we know $|\Agt[t]| \leq k+1$, so we only have to show that
$|\Rgt[t]| = \BO{k+H}$. Let $A^\gamma_{\rm all}(t)$ denote the set of all attractors ever arrived up to time $t$. 
\edit{If $A^\gamma_{\rm all}(t)\setminus \Agt[t]=\emptyset$, then no attractor has been removed, and hence $R_{\rm orph}^{\gamma,t}=\emptyset$. Therefore, $|R^{\gamma,t}|=|\Agt[t]|\leq k+1$, and the claim follows. 
Otherwise, let $q$ be the attractor with largest arrival time among those not in $\Agt[t]$, namely, }
$q = \argmax_{a \in A^\gamma_{\rm all}(t) \setminus \Agt[t]} \tarr(a)$.
Hence, $q$  arrived at some time $\tarr(q) < t$ and was
removed from $\Ag$ at some time $t'$, with
$\tarr(q) \leq t' < t$,
either by Line~\ref{update-step} of the $\textsc{Update}$ procedure, 
because it expired, or by Line~\ref{cleanup-step}
of the \textsc{Cleanup} procedure.
In the time interval $(\tarr(q),t]$ no more than 
$k+1$ new attractors arrived since, otherwise, one of the
attractors arrived in $(\tarr(q),t]$ would not belong to
$\Agt[t]$, because $\Ag$ can never contain more than $k+1$ points,  
and this would contradict the choice of $q$. 
For a point $p \in \Rgt[t]$, let $\psi^\gamma(p)$ denote the attractor
for which $p$ is representative. We partition $\Rgt[t]$ into 3 subsets:
\vspace{-0.25cm}
\begin{itemize}
\item $\Rgt[t](1) = \{ p \in \Rgt[t] :\ \tarr(\psi^\gamma(p)) \in  (\tarr(q), t] \ \}$
\item $\Rgt[t](2) = \{ p \in \Rgt[t] :\ \psi^\gamma(p) \in A^{\gamma,\tarr(q)} \ \}$
\item $\Rgt[t](3) = \{ p \in \Rgt[t] :\ \psi^\gamma(p) \not\in A^{\gamma,\tarr(q)} \text{ and } \tarr(\psi^\gamma(p)) \leq \tarr(q) \ \}$
\end{itemize}
\vspace{-0.25cm}
It is immediate to see that $\Rgt[t]=\Rgt[t](1) \cup \Rgt[t](2) \cup \Rgt[t](3)$. Moreover,  
$|\Rgt[t](1)| \leq k+1$ from what we observed before, 
and $|\Rgt[t](2)| \leq k+1$ since $|A^{\gamma,\tarr(q)}| \leq k+1$ by construction. 
We now argue that
$|\Rgt[t](3)| = \BO{H}$. 
First note that, since arrival times are distinct, at most one point of $R^{\gamma,t}(3)$ can have $t_{arr}(\psi^\gamma(p)) = t_{arr}(q)$; all remaining points have attractors that entered and left $A^\gamma$ before $t_{arr}(q)$, and thus arrived prior to $t_{arr}(q)$.
Let $p$ be the point of $\Rgt[t](3)$ with earliest arrival time. So, all other points of $\Rgt[t](3)$ must have arrived
in the open interval $(\tarr(p),\tarr(q))$. Thus, to show that $|\Rgt[t](3)| = \BO{H}$ it is sufficient to argue that less than $H$ points arrived in
$(\tarr(p),\tarr(q))$. By contradiction, suppose that this is not the case. Then, the assumption that $X$ is $H$-ordered implies that $\tdel(p) < \tdel(q)$. Since $q$ 
was removed from $\Ag$ at time $t' < t$, we have two cases:
\vspace{-0.2cm}
\begin{itemize}
\item
$q$ expired at $t'$ and was removed
at Line~\ref{update-step} of the $\textsc{Update}$ procedure.
Hence,  $\tdel(p) < \tdel(q) = t'<t$, and $p$ cannot be in $\Rgt[t]$. 
\vspace{-0.1cm}
\item
$q$ was removed from $\Ag$ at time $t'$ at Line~\ref{cleanup-step} 
of the \textsc{Cleanup} procedure, since $q$ had
minimum deletion time among $k+2$ attractors. Then, right after the
removal of $q$ from $\Ag$, $t_{\rm min}$ was set to a value $\geq
\tdel(q)$ and all representatives $r$ with $\tdel(r) < \tdel(q) \leq t_{\rm min}$ were removed.  Therefore, since $\tdel(p) < \tdel(q)$,
then $p$ would have been removed as well.
\end{itemize}
\vspace{-0.1cm}
In either case, $p$ could not be in $\Rgt[t]$, which contradicts the choice of $p$.  
\end{proof}

\begin{proof}[Proof of Theorem~\ref{thm:upd-6}]
    Consider a $\gamma \in \Gamma$. The removal of expired points performed in  Line~\ref{update-step}, can be done with a linear scan of in $O(|\Agt[{\rm pr}(\tarr(p))]|+|\Rgt[{\rm pr}(\tarr(p))]|)$, as well as creating the set $E$, and potentially the set $E'$. Adding points to $A^\gamma$ and $R^\gamma$, setting $rep^\gamma(p)$ to $p$, and selecting an arbitrary point of $E'$ can be done in constant time. Finally, procedure $\textsc{Cleanup}$ can be executed in 
    $O(|\Agt[{\rm pr}(\tarr(p))]|+|\Rgt[{\rm pr}(\tarr(p))]|)$ time. 
    Recall that, from Lemma~\ref{lem:space-bound-6}, we have that 
    $|\Agt[{\rm pr}(\tarr(p))]|+|\Rgt[{\rm pr}(\tarr(p))]| \in O(\min\{(k+H),|X_{{\rm pr}(\tarr(p))}|)\}$.
    Since, $|\Gamma| \in O(\log\Delta/\eps)$, we have the claim.
\end{proof}

\begin{proof}[Proof of Theorem~\ref{thm:query-6}]
    The removal of expired points at the beginning of the procedure
    $O(|\Agt[{\rm pr}(t)]|+|\Rgt[{\rm pr}(t)]|)$, for each of the $|\Gamma| \in O(\log\Delta/\eps)$ guesses. 
    Then, for each $\gamma$, running the procedure of \citet{hochbaum1986unified} on $R^\gamma$ takes $\BO{k \cdot |R^\gamma|} = \BO{k \cdot |\Rgt[{\rm pr}(t)]|}$ time. 
    Recalling that $|\Agt[{\rm pr}(t)]|+|\Rgt[{\rm pr}(t)]| \in \BO{\min\{(k+H),|X_{{\rm pr}(t)}|\}}$ by Lemma~\ref{lem:space-bound-6} and that $|\Gamma| \in O(\log\Delta/\eps)$, we have the claim.
\end{proof}

\begin{proof}[Proof of Theorem~\ref{thm:mem-6}]
    The claim follows by bounding the number of points maintained for each guess $\gamma$ via Lemma~\ref{lem:space-bound-6}, and by recalling that $|\Gamma| \in O(\log\Delta/\eps)$.
\end{proof}

\end{document}